\documentclass[conference, 10pt]{IEEEtran}
\usepackage{amsthm}
\usepackage{amsfonts}
\usepackage{amssymb}
\usepackage{mathrsfs}
\usepackage{mathtools}
\usepackage{float}
\usepackage[pdftex]{graphicx}
\usepackage{amsmath}
\usepackage{color}
\usepackage{multirow,multicol}
\usepackage{graphicx}
\usepackage{times}
\usepackage{textcomp}
\usepackage{verbatim}
\usepackage[table]{xcolor}
\usepackage{balance}
\usepackage{lipsum}
\usepackage[inline]{enumitem}
\usepackage{cuted}
\usepackage[caption=false,font=footnotesize]{subfig}
\usepackage{cite}
\usepackage{algpseudocode,algorithm}
\usepackage{hyperref}
\usepackage{tcolorbox}
\usepackage{setspace}

\usepackage[table]{xcolor}
\definecolor{color1}{RGB}{199,209,232}
\definecolor{color2}{RGB}{230,231,233}

\DeclareMathOperator*{\minimize}{minimize} 
\DeclareMathOperator*{\subjectto}{subject\hspace{3pt} to:\hspace{3pt}} 
\newtheorem{theorem}{Theorem}

\newtheorem{proposition}[theorem]{Proposition}

\begin{document}
	
	\title{Millimeter-Wave Radar Beamforming with Spatial Path Index Modulation Communications}
	
	\author{
		\IEEEauthorblockA{Ahmet M. Elbir$^{\dag}$, Kumar Vijay Mishra$^{\ddag}$,  Abdulkadir \c{C}elik$^{+}$, Ahmed M. Eltawil$^{+}$}
		\IEEEauthorblockA{
			${\dag}$Interdisciplinary Centre for Security, Reliability and Trust, University of Luxembourg \\
			${\ddag}$United States DEVCOM Army Research Laboratory, Adelphi, USA\\
			$+$King Abdullah University of Science and Technology, Saudi Arabia}
		\IEEEauthorblockA{E-mail: ahmetmelbir@gmail.com, kvm@ieee.org, abdulkadir.celik@kaust.edu.sa, ahmed.eltawil@kaust.edu.sa }
	}
	
	\maketitle


	\begin{abstract}
		To efficiently utilize the wireless spectrum and save hardware costs, the fifth generation and beyond (B5G) wireless networks envisage integrated sensing and communications (ISAC) paradigms to jointly access the spectrum. In B5G systems, the expensive hardware is usually avoided by employing hybrid beamformers that employ fewer radio-frequency chains but at the cost of the multiplexing gain. Recently, it has been proposed to overcome this shortcoming of millimeter wave (mmWave) hybrid beamformers through spatial path index modulation (SPIM), which modulates the spatial paths between the base station and users and improves spectral efficiency. In this paper, we propose an SPIM-ISAC approach for hybrid beamforming to simultaneously generate beams toward both radar targets and communications users. We introduce a low complexity approach for the design of hybrid beamformers, which include radar-only and communications-only beamformers. Numerical experiments demonstrate that our SPIM-ISAC approach exhibits a significant performance improvement over the conventional mmWave-ISAC design in terms of spectral efficiency and the generated beampattern.
	\end{abstract}


	\section{Introduction}
	\label{sec:Introduciton}
	Radar and communications systems have witnessed tremendous progress for several decades while exclusively operating in different frequency bands to minimize the interference to each other~\cite{mishra2019toward}. Modern radar systems operate in various portions of the spectrum -- from very-high-frequency (VHF) to Terahertz (THz)~\cite{elbir_thz_jrc_Elbir2022Aug} -- for different applications, such as over-the-horizon, air surveillance, meteorological, military, and automotive radars. Similarly, communications systems are progressing from ultra-high-frequency (UHF) to millimeter-wave (mmWave) in response to the demand for new services, massive number of users, and high data rate requirements for the applications
	\cite{elbir2021JointRadarComm}. As a result, there has been substantial interest in designing \textit{integrated sensing and communications} (ISAC) to jointly access the spectrum~\cite{mishra2019toward,liu2020co,duggal2020doppler,sedighi2021localization}.

	Future wireless communications are also focused on technologies to improve improved energy/spectral efficiency (EE/SE)~\cite{heath2016overview}.  In this context, index modulation (IM)  is emerging as an attractive area of research because it offers both better EE and SE over conventional modulators~\cite{indexMod_Survey_Mao2018Jul}. In IM, the transmitter encodes additional information in the indices of the transmission media such as subcarriers~\cite{jrc_IndexMod_Huang2020May,hodge2020intelligent}, antennas~\cite{antenna_grouping_SM,hodge2019reconfigurable}, and spatial paths~\cite{spim_bounds_JSTSP_Wang2019May,spim_BIM_TVT_Ding2018Mar,spim_GBM_Gao2019Jul}. In this paper, we focus on spatial modulation (SM) in the context of mmWave multiple-input multiple-output (MIMO) ISAC systems~\cite{mishra2019toward,elbir_thz_jrc_Elbir2022Aug}.
	
	In the mmWave-MIMO, hybrid analog/digital beamformers are employed, where the number of radio-frequency (RF) chains is much smaller than the antennas. While this saves cost and power, its multiplexing gain is limited~\cite{heath2016overview,elbir_Beamforming25_Elbir2022Nov}. The SM techniques have been shown to be helpful in addressing this problem. For SM-aided ISAC systems, \cite{jrc_spim_sm_Ma2021Feb} devised an SM approach, wherein the IM is performed over the antenna indices to handle sensing and communications tasks jointly. In~\cite{jrc_spim_SpatialPrthogonal_Li2022Mar}, a spatially orthogonal time-frequency SM was suggested for ISAC applications. This was further investigated in~\cite{jrc_generalized_SM_Xu2020Sep}, which employed SM over antenna indices for orthogonal frequency division multiplexing (OFDM) ISAC with each subcarrier assigned exclusively to an active antenna. 
	
	A more generalized scenario, i.e., spatial path index modulation (SPIM) was considered in~\cite{spim_BIM_TVT_Ding2018Mar}. Here, the indices of the spatial paths were modulated to create different \emph{spatial patterns} for mmWave-MIMO. In~\cite{spim_GBM_Gao2019Jul}, beamspace modulation was exploited by employing lens arrays at both transmitter and receiver. Moreover, SE was utilized as a performance metric in~\cite{spim_bounds_JSTSP_Wang2019May} for analog-only beamforming. A distributed machine  learning approach was used in~\cite{spim_FL_Elbir2021Jun} for multi-user SPIM in mmWave-MIMO. To sum up, the aforementioned ISAC works \cite{jrc_spim_sm_Ma2021Feb,jrc_spim_SpatialPrthogonal_Li2022Mar,jrc_generalized_SM_Xu2020Sep} do not exploit the SPIM while they only consider SM over antennas indices. Furthermore, the proposed SPIM approaches~\cite{spim_BIM_TVT_Ding2018Mar,spim_GBM_Gao2019Jul,spim_bounds_JSTSP_Wang2019May,spim_FL_Elbir2021Jun} consider the communications-only scenario without considering the trade-off between the radar and communications functionalities.

	In this paper, we introduce an SPIM-based hybrid beamformer design approach for ISAC. 
	Our proposed beamformer simultaneously maximizes the spectral efficiency at the communications user over SPIM-aided signaling and achieves as much signal-to-noise ratio (SNR) as possible for detecting the radar target. The SPIM-ISAC analog beamformer comprises radar-only and communications-only beamformers, which are selected from different spatial patterns between the base station (BS) and the communications user. The proposed design also includes a trade-off parameter between communications and radar sensing operations in the sense that the SNR at the targets and users is controlled.  In order to design the radar-only beamformer, the target direction is estimated in the search phase of the radar~\cite{music,elbir_DL_MUSIC}. Then, all possible spatial patterns are exploited from the estimated channel obtained via limited feedback techniques~\cite{mimoRHeath,heath2016overview}. Once the radar- (e.g., target directions) and communication- (e.g., path directions) related parameters are collected at the BS, the hybrid beamformer is designed for each spatial pattern in accordance with the trade-off parameter. The effectiveness of the proposed SPIM-ISAC approach is evaluated via numerical experiments and compared with conventional mmWave-ISAC, whose beamformers are designed in accordance with the strongest path between the BS and the user. We have shown that a significant performance improvement is achieved by the proposed approach in terms of communications- and radar-related  performance metrics.

	\textit{Notation:} Throughout the paper,  $(\cdot)^\textsf{T}$ and $(\cdot)^{\textsf{H}}$ denote the transpose and conjugate transpose operations, respectively. For a matrix $\mathbf{A}$ and vector $\mathbf{a}$; $[\mathbf{A}]_{ij}$, $[\mathbf{A}]_k$  and $[\mathbf{a}]_l$ correspond to the $(i,j)$-th entry, $k$-th column and $l$-th entry, respectively. $\lfloor\cdot \rfloor$ and $\mathbb{E}\{\cdot\}$ represent the flooring and expectation operations, respectively. We denote $|| \cdot||_2$ and $|| \cdot||_\mathcal{F}$ as the  $l_2$-norm and Frobenious norm, respectively.

	\section{System Model}
	Consider the transmitter design problem in an ISAC scenario involving a communications user and a radar target with SPIM (Fig.~\ref{fig_BS}). The BS has $N_\mathrm{T}$ antennas to jointly communicate with the user and sense the target via probing signals. The user has $N_\mathrm{R}$ antennas, for which $N_\mathrm{S}$ data symbols $\mathbf{s} = [s_1,\cdots,s_{N_\mathrm{S}}]^\textsf{T}\in \mathbb{C}^{N_\mathrm{S}}$ are transmitted, where $\mathbb{E}\{\mathbf{ss}^\textsf{H}\}=\mathbf{I}_{N_\mathrm{S}}$.  Additionally, the spatial path index information represented by ${s}_0$ is fed to the switching network (Fig.~\ref{fig_BS}) to randomly assign the outputs of $N_\mathrm{RF}$ RF chains to the $ \bar{M}$ taps of the analog beamformer. Thus, the BS processes at most $ \bar{M} = M+1$ spatial paths, where $M$ denotes the number of available spatial paths for the user while a single path is dedicated to the target. Note that the proposed system design is also applicable to multiple targets.
	
	Compared to the conventional mmWave systems, SPIM has the advantage of transmitting additional data streams by exploiting the \emph{spatial pattern} of the mmWave channel with limited RF chains, i.e., $N_\mathrm{RF} \leq  \bar{M}$~\cite{spim_bounds_JSTSP_Wang2019May}. For example, if $ {M}=1$, i.e., $N_\mathrm{RF} =2$, then SPIM reduces to conventional mmWave ISAC design because there is only one choice of transmission, i.e., one RF chain dedicated for radar target and the another one for the communication user. Define the  total number of spatial patterns as \cite{spim_GBM_Gao2019Jul}
	\begin{align}
	K = 2^{\left\lfloor \log_2 \left(\footnotesize\begin{array}{c}
		M \\
		\bar{N}_\mathrm{RF}
		\end{array}  \right) \right \rfloor  },
	\end{align}
	where $\bar{N}_\mathrm{RF} = N_\mathrm{RF}-1$, which indicates that the first column of the analog beamformer is dedicated to radar sensing while the remaining columns are employed for communications. Denote the index set of possible spatial patterns by $\mathcal{I} = \{1,\cdots, K\}$. Then, the $N_\mathrm{T}\times 1$ transmit signal for the $i$-th,  $i\in \mathcal{I}$, spatial pattern becomes
	\begin{align}
	\mathbf{x}^{(i)} = \mathbf{F}_\mathrm{RF}^{(i)}\mathbf{F}_\mathrm{BB}\mathbf{s},
	\end{align}
	where $\mathbf{F}_\mathrm{BB}\in\mathbb{C}^{N_\mathrm{RF}\times N_\mathrm{S}}$ is the baseband beamformer, and   ${\mathbf{F}}_\mathrm{RF}^{(i)} = [\mathbf{f}_\mathrm{R}, \tilde{\mathbf{F}}_\mathrm{RF}^{(i)}]\in N_\mathrm{T}\times N_\mathrm{RF}$ is the analog beamformer matrix containing the radar-only beamformer $\mathbf{f}_\mathrm{R}\in\mathbb{C}^{N_\mathrm{T}}$ and $\tilde{\mathbf{F}}_\mathrm{RF}^{(i)}\in \mathbb{C}^{N_\mathrm{T}\times (N_\mathrm{RF}-1)}$ is the communications-only analog beamformer for the $i$-th spatial pattern. Note that the analog beamformer $\mathbf{F}_\mathrm{RF}^{(i)}$ has constant-modulus constraint, i.e., $|[\mathbf{F}_\mathrm{RF}^{(i)}]_{i,j}| = 1/\sqrt{N_\mathrm{T}}$ for $i = 1,\cdots, N_\mathrm{T}$, $j = 1,\cdots,N_\mathrm{RF}$.

	\begin{figure}[t]
		\centering		{\includegraphics[draft=false,width=\columnwidth]{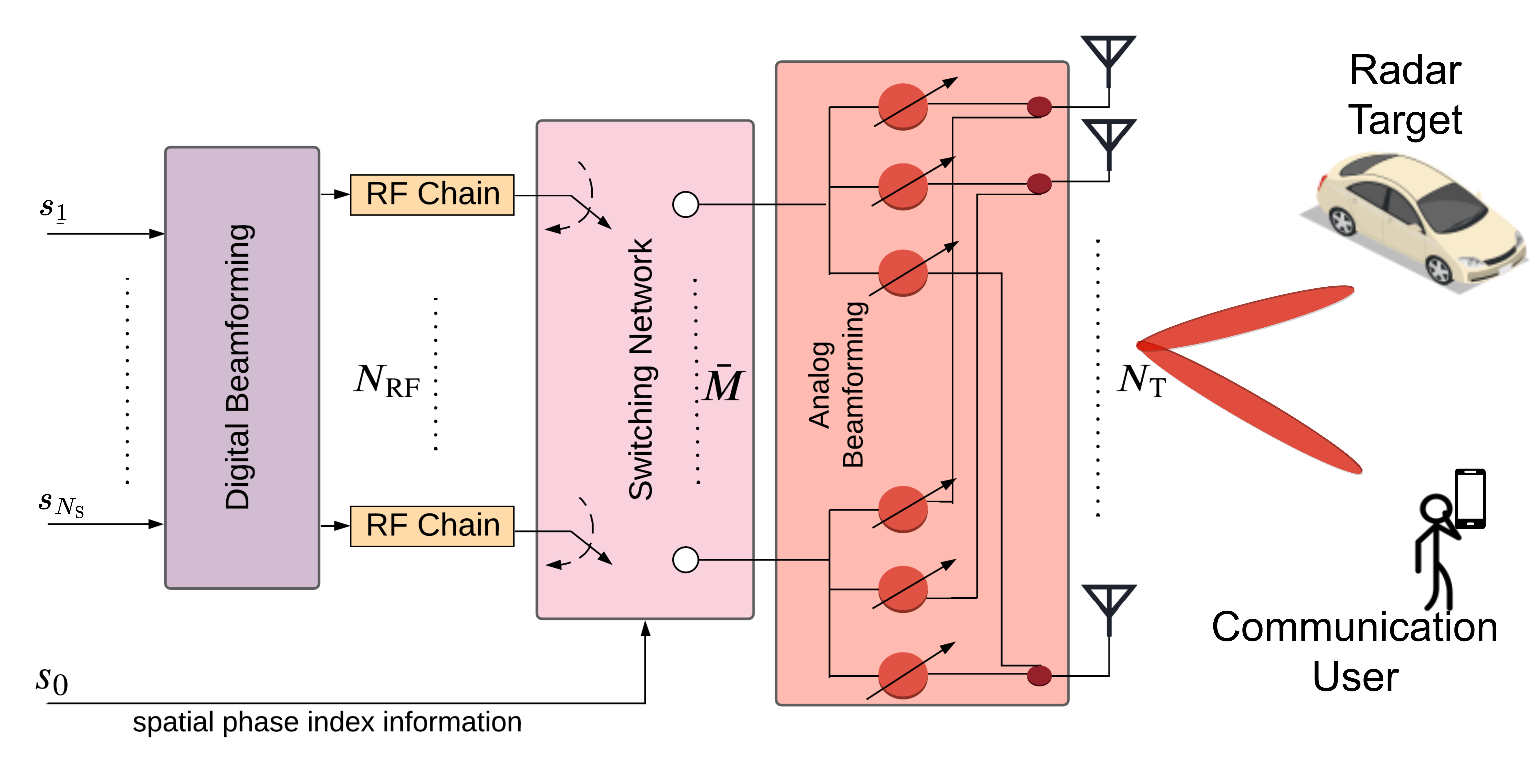} }
		\vspace*{-5mm} 
		\caption{The SPIM-ISAC architecture processes the incoming data streams and employs spatial path index information $s_0$ in a switching network, which connects $N_\mathrm{RF}$ RF chains to $\bar{M} = M+1$ taps on the analog beamformers to exploit a single path for the radar target and one of the $M$ spatial paths for the communications user. 
		}
		\label{fig_BS}
	\end{figure}

	\subsection{Communications Receiver}
	Denote the direction-of-arrival (DoA) and direction-of-departure (DoD) angles of the scattering paths between the user and the BS by $\phi$ and $\theta$, respectively. The corresponding receive and transmit steering vectors is defined as $\mathbf{a}_\mathrm{R}(\phi)\in \mathbb{C}^{N_\mathrm{R}}$ and $\mathbf{a}_\mathrm{T}(\theta)\in \mathbb{C}^{N_\mathrm{T}}$, respectively. Then, for a uniform linear array (ULA) (with half-wavelength element spacing), the $n$-th elements of $\mathbf{a}_\mathrm{R}(\phi)$ and $\mathbf{a}_\mathrm{T}(\theta)$ are
	\begin{align}
	[\mathbf{a}_\mathrm{R}( \phi)]_n  &= \frac{1}{\sqrt{N_\mathrm{R}}}\exp\{-\mathrm{j}\pi (n-1) \sin (\phi) \} \nonumber\\
	[\mathbf{a}_\mathrm{T}( \theta)]_n &=\frac{1}{\sqrt{N_\mathrm{T}}}\exp\{-\mathrm{j}\pi (n-1) \sin (\theta) \}.
	\end{align}
	Then, the $N_\mathrm{R}\times N_\mathrm{T}$ mmWave channel is 
	\begin{align}
	\label{channel1}
	\mathbf{H} = \sum_{m = 1}^{M} \bar{\gamma}_m \mathbf{a}_\mathrm{R}(\phi_m) \mathbf{a}_\mathrm{T}^\textsf{H}(\theta_m),
	\end{align}
	where $\bar{\gamma}_m\in \mathbb{C}$ denotes the channel path gains for $m = 1,\cdots, M$. 
	
	In a compact form, the channel expression in (\ref{channel1}) is~\cite{heath2016overview}
	\begin{align}
	\label{channel}
	\mathbf{H} = \mathbf{P}\boldsymbol{\Lambda} \mathbf{Q}^\textsf{H},
	\end{align}
	where the matrices $\mathbf{P} \in \mathbb{C}^{N_\mathrm{R}\times M}$ and $\mathbf{Q}\in \mathbb{C}^{N_\mathrm{T}\times M}$ represent the receive and transmit array responses for $M$ paths, respectively, and $\boldsymbol{\Lambda}\in \mathbb{C}^{M \times M}$ is a diagonal matrix comprised of path gains $\bar{\gamma}_m = \sqrt{\gamma_m}$ as
	\begin{align}
	\boldsymbol{\Lambda} = \mathrm{diag}\{[\sqrt{\gamma_{1}},\cdots, \sqrt{\gamma_M}]\},
	\end{align}
	where $\gamma_{1} > \gamma_{2} > \cdots > \gamma_M$; in addition, we have $\mathbf{P} = [\mathbf{a}_\mathrm{R}(\phi_1),\cdots, \mathbf{a}_\mathrm{R}(\phi_M)]$ and $\mathbf{Q} = [\mathbf{a}_\mathrm{T}(\theta_1),\cdots,\mathbf{a}_\mathrm{T}(\theta_M)]$.
	
	Then, the $N_\mathrm{R}\times 1$ received signal at the communication user for the  $i$th spatial pattern is
	\begin{align}
	\mathbf{y}^{(i)} = \mathbf{H}\mathbf{F}_\mathrm{RF}^{(i)}\mathbf{F}_\mathrm{BB}\mathbf{s} + \mathbf{n},
	\end{align}
	where $\mathbf{n}\sim \mathcal{CN}(\mathbf{0},\sigma_n^2\mathbf{I}_{N_\mathrm{R}})\in\mathbb{C}^{N_\mathrm{R}}$ represents the temporarily and spatially  additive white Gaussian noise vector.

	\subsection{Radar Receiver}
	The aim of radar processing is to achieve the highest possible  SNR gain toward the direction of the target. Denote the estimate of radar target direction by $\Phi$ and select the radar-only beamformer as $\mathbf{f}_\mathrm{R} = \mathbf{a}_\mathrm{T}(\Phi)$. Then, using the hybrid beamforming structure, the beampattern of the radar is
	\begin{align}
	B^{(i)}(\Phi) = \mathrm{Trace}\{\mathbf{Q}^\textsf{H}(\Phi)\mathbf{R}_\mathbf{x}^{(i)} \mathbf{Q}(\Phi)  \},
	\end{align}
	where $\mathbf{R}_\mathbf{x}^{(i)}\in\mathbb{C}^{N_\mathrm{T}\times N_\mathrm{T}}$ is the covariance matrix of the transmit signal of the ISAC transmitter subject to the hybrid architecture of the beamformers. For the $i$-th spatial pattern, we have 
	\begin{align}
	\mathbf{R}_\mathbf{x}^{(i)} &= \mathbb{E}\{ \mathbf{xx}^\textsf{H} \}=\mathbb{E}\{ \mathbf{F}_\mathrm{RF}^{(i)}\mathbf{F}_\mathrm{BB}\mathbf{ss}^\textsf{H} \mathbf{F}_\mathrm{BB}^\textsf{H}\mathbf{F}_\mathrm{RF}^{(i)^\textsf{H}} \} \nonumber \\
	& = \mathbf{F}_\mathrm{RF}^{(i)}\mathbf{F}_\mathrm{BB}\mathbb{E}\{\mathbf{ss}^\textsf{H}\} \mathbf{F}_\mathrm{BB}^\textsf{H}\mathbf{F}_\mathrm{RF}^{(i)^\textsf{H}} \nonumber \\
	&= \mathbf{F}_\mathrm{RF}^{(i)}\mathbf{F}_\mathrm{BB} \mathbf{F}_\mathrm{BB}^\textsf{H}\mathbf{F}_\mathrm{RF}^{(i)^\textsf{H}}.
	\end{align}
	To simultaneously obtain the desired beampattern for the radar target and provide satisfactory communication performance, the hybrid beamformer $\mathbf{F}_\mathrm{RF}^{(i)}\mathbf{F}_\mathrm{BB}$ should be designed accordingly, as discussed in the following.
	
	
	\subsection{Problem Formulation}
	For ISAC beamformer design, $\mathbf{F}_\mathrm{RF}^{(i)}$, $\mathbf{F}_\mathrm{BB}$, we minimize the Euclidean distance between the ISAC beamformers and the unconstrained beamformers. The fully-digital unconstrained communication-only beamformer is  $\mathbf{F}_\mathrm{opt}\in\mathbb{C}^{N_\mathrm{T}\times N_\mathrm{S}}$, which is obtained from the singular value decomposition (SVD) of $\mathbf{H}$~\cite{heath2016overview}. Define the joint radar-communications beamformer~\cite{elbir2021JointRadarComm} as
	\begin{align}
	\label{Fcr}
	\mathbf{F}_\mathrm{CR} = \eta \mathbf{F}_\mathrm{opt} + (1- \eta) \mathbf{f}_\mathrm{R}\boldsymbol{\xi},
	\end{align} 
	where $\boldsymbol{\xi}\in\mathbb{C}^{1\times N_\mathrm{S}}$ row vector providing the change of dimensions between $\mathbf{f}_\mathrm{R}$ and $\mathbf{F}_\mathrm{CR}$. In (\ref{Fcr}), $0\leq \eta\leq 1$ provides the trade-off between the radar and communications tasks. In particular, $\eta=1$ ($\eta = 0$) corresponds to the communications-only (radar-only) design problem.
	
	By combining communications-only and radar-only designs, the joint problem becomes
	\begin{align}
	\label{prob1}
	\minimize_{\mathbf{F}_\mathrm{RF}^{(i)}, \mathbf{F}_\mathrm{BB},\boldsymbol{\xi}} & \| \mathbf{F}_\mathrm{CR} - \mathbf{F}_\mathrm{RF}^{(i)}\mathbf{F}_\mathrm{BB} \|_\mathcal{F} \nonumber\\
	\subjectto &\hspace{20pt}|[\mathbf{F}_\mathrm{RF}^{(i)}]_{n,r}| = 1/\sqrt{N_\mathrm{T}},\nonumber\\ &\tilde{\mathbf{F}}_\mathrm{RF}^{(i)}\in \mathcal{A},\nonumber\\
	&\|\mathbf{F}_\mathrm{RF}^{(i)}\mathbf{F}_\mathrm{BB}\|_\mathcal{F} = N_\mathrm{S} ,\nonumber\\
	&\boldsymbol{\xi} \boldsymbol{\xi}^\textsf{H} = 1 ,
	\end{align}
	where $\mathcal{A}$ denotes the set of possible analog beamformers for the SPIM.  
	
	The optimization problem in \eqref{prob1} falls to the class of mixed-integer non-convex programming 
	and computationally prohibitive due to combinatorial subproblem for each spatial pattern, and non-linear due to  multiple unknowns $\mathbf{F}_\mathrm{RF}^{(i)}, \mathbf{F}_\mathrm{BB},\boldsymbol{\xi}$. Instead, we propose a low complexity approach by exploiting the steering vectors corresponding to the path directions in the following.

	\section{SPIM in ISAC}
	The design of the analog beamformer $\mathbf{F}_\mathrm{RF}^{(i)}$ requires the knowledge of radar-only beamformer $\mathbf{f}_\mathrm{R}$ and communications-only beamformers $\tilde{\mathbf{F}}_\mathrm{RF}^{(i)}$ for $i\in \mathcal{I} = \{1,\cdots, K\}$. In other words, $\bar{N}_\mathrm{RF}$ columns of $N_\mathrm{T}\times N_\mathrm{RF}$ analog beamformer $\mathbf{F}_\mathrm{RF}$ are dedicated to the communications task, while a single column, i.e., $\mathbf{f}_\mathrm{R}$ is dedicated to the radar operations.
	
	\subsection{Hybrid Beamformer Design}
	The ISAC hybrid beamformer is comprised of radar-only and communication-only beamformers. In order to obtain the radar-only beamformer $\mathbf{f}_\mathrm{R}$, the target direction $\Phi$ is estimated in the search phase of the radar. Then, the $\mathbf{f}_\mathrm{R}$ is constructed as the steering vector corresponding to $\Phi$~\cite{elbir2021JointRadarComm,radarCommLiuICASSP2019}. To this end, the BS first transmits probing signals, for which the $N_\mathrm{T}\times 1$ received array output at the BS is 
	\begin{align}
	\label{receivedRadar}
	\bar{\mathbf{y}} (t) = \mathbf{a}_\mathrm{T}(\Phi)\mathbf{a}_\mathrm{T}^\textsf{T}(\Phi) r(t) + \bar{\mathbf{n}}(t),
	\end{align}
	where $t$ denotes the sample index for $t = 1,\cdots, T_R$, where $T_R$ is the number of snapshots, $r(t)$ is the reflection coefficient from the target at direction $\Phi$ and $\bar{\mathbf{n}}\in \mathbb{C}^{N_\mathrm{T}}$ is the noise term.  Next, the BS utilizes the sample covariance matrix of the received signal in (\ref{receivedRadar}) as
	\begin{align}
	\bar{\mathbf{R}}_\mathbf{y} = \frac{1}{T_R} \sum_{t = 1}^{T_R} \bar{\mathbf{y}}(t) \bar{\mathbf{y}}^\textsf{H}(t),
	\end{align}
	which is used to estimate $\Phi$ via both model-based techniques, e.g., \textit{mu}ltiple \textit{si}gnal \textit{c}lassification (MUSIC)~\cite{music} and model-free, deep learning based approaches, e.g., DeepMUSIC~\cite{elbir_DL_MUSIC}. Then, the radar-only beamformer is selected as $\mathbf{f}_\mathrm{R} = \mathbf{a}_\mathrm{T}(\hat{\Phi})$.
	
	The design of communications-only analog beamformer $\tilde{\mathbf{F}}_\mathrm{RF}^{(i)}$ requires knowledge of all scattering paths, which constitute the array response matrix $\mathbf{Q}$ given in \eqref{channel}. In fact, the  channel parameters $\mathbf{P}$, $\boldsymbol{\Lambda}$, and $\mathbf{Q}$ are computed during the channel estimation stage of the receiver via both model-based~\cite{mimoRHeath,heath2016overview} and model-free techniques~\cite{dl_ChannelEst_Abdallah2021Nov,dl_online_CE_HB_Elbir2021Dec}. These estimates are then sent to the BS using limited feedback techniques~\cite{heath2016overview}.

	
	Once the beamformers - radar-only $\mathbf{f}_\mathrm{R}$ and communications-only beamformers $\tilde{\mathbf{F}}_\mathrm{RF}^{(i)}$ - are obtained, then the analog beamformer is designed as
	\begin{align}
	\label{Frf}
	{\mathbf{F}}_\mathrm{RF}^{(i)} = [\mathbf{f}_\mathrm{R}, \tilde{\mathbf{F}}_\mathrm{RF}^{(i)}].
	\end{align}
	Furthermore, the baseband beamformer is $\mathbf{F}_\mathrm{BB} = \mathrm{blkdiag}\{(\eta-1), \eta \mathbf{I}_{N_\mathrm{RF}-1}  \}$, which allows the trade-off between the radar and communications tasks.  The trade-off between the radar and communications tasks is controlled by choosing $\eta$ depending on the accuracy/importance of both tasks in the ISAC system~\cite{elbir_IRS_ISAC_Elbir2022Apr}.


	
	\subsection{Performance Analysis}
	The radar sensing performance of the proposed SPIM-ISAC system is quantified by computing the beampattern of the hybrid beamformers (see Sec.~\ref{sec:RadarSensingSim}). In order to evaluate the communication performance, the mutual information (MI) between the transmit and receive signals, i.e.,  $\mathbf{x}^{(i)}$, $\mathbf{y}^{(i)}$ is utilized.  In the following, we compute the MI expressions for mmWave-ISAC and SPIM-ISAC, respectively.
	
	
	Denote arbitrary hybrid beamformers by ${\mathbf{F}}_\mathrm{RF}$ and ${\mathbf{F}}_\mathrm{BB}$. The MI is computed as~\cite{heath2016overview}
	\begin{align}
	\label{MI_general}
	\mathcal{M} = \log_2 \left(\mathrm{det}\left\{\mathbf{I}_{N_\mathrm{R}} + \frac{1}{\sigma_n^2}\mathbf{H}{\mathbf{F}}_\mathrm{RF}{\mathbf{F}}_\mathrm{BB}{\mathbf{F}}_\mathrm{BB}^\textsf{H}{\mathbf{F}}_\mathrm{RF}^\textsf{H}\mathbf{H}^\textsf{H}   \right \}  \right).
	\end{align}
	
	In the conventional mmWave systems, the analog beamformer $\tilde{\mathbf{F}}_\mathrm{RF}$ relies on the selection of the strongest path for hybrid beamformer design~\cite{spim_AsymptoticMI_He2017Nov,spim_bounds_JSTSP_Wang2019May}. As an example, for a single target and single user case, we have  $\mathbf{F}_\mathrm{RF} =[ \mathbf{a}_\mathrm{T}(\Phi),\mathbf{a}_\mathrm{T}(\theta_1)]$, where $\mathbf{a}_\mathrm{T}(\theta_1)$ corresponds to the strongest communication path with path gain $\gamma_{1}$. Similarly, the MI for mmWave-ISAC system is computed for the spatial pattern $i=1$ as
	\begin{align}
	&\mathcal{M}_\mathrm{mmWave}  =\log_2 \bigg(\mathrm{det}\bigg\{\mathbf{I}_{N_\mathrm{R}} \nonumber \\
	&\hspace{70pt}+ \frac{1}{\sigma_n^2}\mathbf{H}{\mathbf{F}}_\mathrm{RF}^{(1)}{\mathbf{F}}_\mathrm{BB}{\mathbf{F}}_\mathrm{BB}^\textsf{H}{\mathbf{F}}_\mathrm{RF}^{(1)^\textsf{H}}\mathbf{H}^\textsf{H}   \bigg \}  \bigg). \label{MI_mmwave1}
	\end{align}

	The computation of MI in (\ref{MI_mmwave1}) is computationally complex, especially for large number of antennas.  In the following Proposition~\ref{prop:mi}, we introduce a closed-form expression for the asymptotic MI  of the  mmWave-ISAC system with massive antenna array assumption, i.e., $N_\mathrm{T}\gg 1$.

	
	
	\begin{proposition}\label{prop:mi}
		Consider the mmWave-ISAC system with massive antenna array deployment (i.e., $N_\mathrm{T}\gg 1$). Then, the MI for conventional mmWave-ISAC system is 
		\begin{align}
		\label{mmWaveTheoretical}
		\bar{\mathcal{M}}_\mathrm{mmWave} = \log_2 \left( 1 +  \frac{\eta^2\gamma_1 }{\sigma_n^2} \right),
		\end{align}
		where $\gamma_1$  corresponds to the strongest path gain.
	\end{proposition}
	
	\begin{proof}
		Since only the strongest path is considered in mmWave-ISAC~\cite{spim_AsymptoticMI_He2017Nov,spim_bounds_JSTSP_Wang2019May}, the analog and digital beamformers are $\mathbf{F}_\mathrm{RF} =\left[ \mathbf{a}_\mathrm{T}(\Phi),\mathbf{a}_\mathrm{T}(\theta_1)\right] $ and $\mathbf{F}_\mathrm{BB}=\left[\begin{array}{cc}
		(1-\eta) &0 \\
		0 & \eta 
		\end{array}\right]$, respectively. Then, using the expression in (\ref{MI_general}, the asymptotic MI for mmWave-ISAC is 
		\begin{align}
		\label{mmWave2}
		&\bar{\mathcal{M}}_\mathrm{mmWave} =  \log_2  \bigg(  \mathrm{det}\bigg\{\mathbf{I}_{N_\mathrm{R}}  +  \frac{1}{\sigma_n^2} \mathbf{H}\left[ \mathbf{a}_\mathrm{T}(\Phi),\mathbf{a}_\mathrm{T}(\theta_1)\right]\nonumber\\
		&	\hspace{20pt}  \times  \left[\begin{array}{cc}
		(1-\eta)^2 &0 \\
		0 & \eta^2 
		\end{array}\right]  \left[ \mathbf{a}_\mathrm{T}(\Phi),\mathbf{a}_\mathrm{T}(\theta_1)\right]^\textsf{H}   \mathbf{H}^\textsf{H}    \bigg\}    \bigg) \nonumber \\
		&=   \log_2  \bigg(  \mathrm{det}\bigg\{\mathbf{I}_{N_\mathrm{R}}   +  \frac{1}{\sigma_n^2} \mathbf{H}  \bigg((1-\eta)^2 \mathbf{a}_\mathrm{T}(\Phi)\mathbf{a}_\mathrm{T}^\textsf{H}(\Phi) \nonumber\\
		&\hspace{30pt} + \eta^2\mathbf{a}_\mathrm{T}(\theta_1)\mathbf{a}_\mathrm{T}^\textsf{H}(\theta_1) \bigg)  \mathbf{H}^\textsf{H}    \bigg\}    \bigg). 
		\end{align}
		Using the orthogonality of steering vectors at different angles, i.e., $\lim_{N_\mathrm{T}\rightarrow +\infty} |\mathbf{a}_\mathrm{T}^\textsf{H}(\theta_1)\mathbf{a}_\mathrm{T}(\Phi)  | = 0  $, we get $\mathbf{H}  \mathbf{a}_\mathrm{T}(\Phi) = \mathbf{0}_{N_\mathrm{R}}$. Then, (\ref{mmWave2}) becomes
		\begin{align}
		\label{mmWave3}
		&\bar{\mathcal{M}}_\mathrm{mmWave} =   \log_2  \bigg(  \mathrm{det}\bigg\{\mathbf{I}_{N_\mathrm{R}}   +  \frac{\eta^2}{\sigma_n^2} \mathbf{H}   \mathbf{a}_\mathrm{T}(\theta_1)\mathbf{a}_\mathrm{T}^\textsf{H}(\theta_1)  \mathbf{H}^\textsf{H}    \bigg\}    \bigg),
		\end{align}
		where we have  $\mathbf{H}\mathbf{a}_\mathrm{T}(\theta_1) = \sqrt{\gamma_{1}}\mathbf{a}_\mathrm{R}(\phi_1) \mathbf{a}_\mathrm{T}^\textsf{H}(\theta_1)\mathbf{a}_\mathrm{T}(\theta_1)= \sqrt{\gamma_{1} } \mathbf{a}_\mathrm{R}(\phi_1) $ with $\mathbf{a}_\mathrm{T}^\textsf{H}(\theta_1) \mathbf{a}_\mathrm{T}(\theta_1) = 1$. Hence, (\ref{mmWave3}) becomes 
		\begin{align}
		\bar{\mathcal{M}}_\mathrm{mmWave}\hspace{-3pt} &= \log_2  \left(  \mathrm{det}\left\{\mathbf{I}_{N_\mathrm{R}}  +  \frac{\eta^2\gamma_{1}}{\sigma_n^2} \mathbf{a}_\mathrm{R}(\phi_1)\mathbf{a}_\mathrm{R}^\textsf{H}(\phi_1)    \right\}     \right) \nonumber \\
		&= \hspace{-3pt}\log_2  \left(\hspace{-3pt}  \mathrm{det}\left\{1  +  \frac{\eta^2\gamma_{1}}{\sigma_n^2} \mathbf{a}_\mathrm{R}^\textsf{H}(\phi_1)\mathbf{a}_\mathrm{R}(\phi_1)  \hspace{-3pt}  \right\}    \right) \label{mmWaveProof1} \\
		&= \log_2  \left(  1  +  \frac{\eta^2\gamma_{1}}{\sigma_n^2}       \right), \label{mmWaveProof}
		\end{align}
		where, the equality in (\ref{mmWaveProof1}) utilizes the matrix determinant property, i.e., $\mathrm{det}\{\mathbf{I}_\mathrm{R} + \mathbf{a}_\mathrm{R}(\phi_1)\mathbf{a}_\mathrm{R}^\textsf{H}(\phi_1)     \} = \mathrm{det}\{1  + \mathbf{a}_\mathrm{R}^\textsf{H}(\phi_1)\mathbf{a}_\mathrm{R}(\phi_1) \}  $. Furthermore, we have  $\mathbf{a}_\mathrm{R}(\phi_1)\mathbf{a}_\mathrm{R}^\textsf{H}(\phi_1) = 1$. 
		This completes the proof.
	\end{proof}
	
	Proposition~\ref{prop:mi} allows us to compute the MI for mmWave-ISAC system with low complexity via (\ref{mmWaveTheoretical}). In the following, we introduce the computation of MI for SPIM-ISAC system. 
	
	The MI expression for SPIM-ISAC is obtained by taking into account the received spatially modulated signal $\mathbf{y}^{(i)}$. Hence, we first compute the covariance matrix of $\mathbf{y}^{(i)}$ as
	\begin{align}
	\boldsymbol{\Sigma}_i = \mathbf{H}\mathbf{F}_\mathrm{RF}^{(i)}\mathbf{F}_\mathrm{BB}\mathbf{F}_\mathrm{BB}^\textsf{H}\mathbf{F}_\mathrm{RF}^{(i)^\textsf{H}}\mathbf{H}^\textsf{H} + \sigma_n^2 \mathbf{I}_{N_\mathrm{R}},
	\end{align}
	which includes the effect of ISAC trade-off parameter $\eta$ in $\mathbf{F}_\mathrm{BB}$. 
	Then, the MI of the overall SPIM-ISAC system is computed using the following asymptotic closed-form expression~\cite{spim_AsymptoticMI_He2017Nov,spim_bounds_JSTSP_Wang2019May}, i.e., 
	\begin{align}
	&\mathcal{M}_\mathrm{SPIM} \nonumber\\
	&= \log_2 \left( \frac{K}{(2\sigma_n^2)^{N_\mathrm{R}}}\right) - \frac{1}{K} \sum_{i = 1}^K \log_2  \left(\sum_{j = 1}^K \mathrm{det}\{\boldsymbol{\Sigma}_{i} + \boldsymbol{\Sigma}_{j}  \}^{-1}\right),
	\end{align}
	where $i,j\in \{1,\cdots, M\}$ denote spatial path index.
	
	
	

	\section{Numerical Experiments}
	\label{sec:Sim}
	We evaluated the performance of our SPIM-ISAC approach with conventional mmWave-ISAC in terms of MI averaged over $500$ Monte Carlo trials. The number of antennas at the BS and the users are $N_\mathrm{T}=128$ and $N_\mathrm{R}=10$, respectively. We select the number of available spatial paths as $M=2$ ($\bar{M} = 3$ and $K=2$), and $N_\mathrm{RF} = N_\mathrm{S} = 2$. The target DoA is $\Phi = 40^\circ$ and path directions are drawn from $[-90^\circ,90^\circ]$ uniformly at random.

	\begin{figure}[t]
		\centering
		{\includegraphics[draft=false,width=\columnwidth]{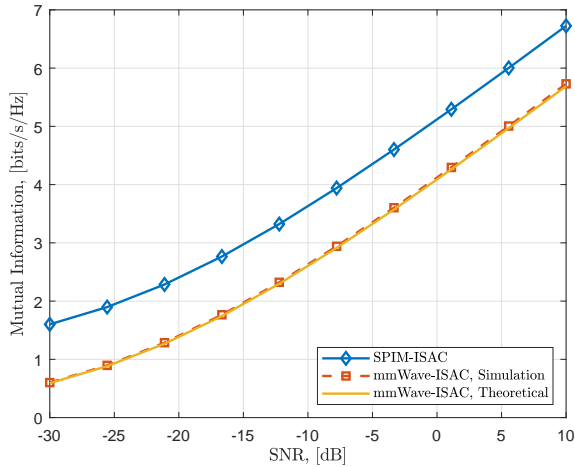} } 
		\vspace*{-6mm}
		\caption{MI versus SNR when the path gains $\gamma_{1} = \gamma_{2} =0.5$ and the radar-communications trade-off parameter $\eta = 0.5$.
		}
		
		\label{fig_SE1}
	\end{figure}
	
	\subsection{Communication Performance}
	
	Fig.~\ref{fig_SE1} shows the MI with respect to SNR (i.e., $1/\sigma_n^2$) when $\gamma_1 = \gamma_{2} =0.5$ and the radar-communications trade-off parameter is $\eta =0.5$. We observe a significant improvement in MI with our proposed SPIM-ISAC approach compared to mmWave-ISAC even only $M=2$ paths are available for SPIM. The computation of mmWave-ISAC is obtained both theoretically (i.e., \eqref{mmWaveTheoretical} ) and numerically (i.e., (\ref{MI_mmwave1})). We observe that both simulated and theoretical curves match as shown in Fig.~\ref{fig_SE1}.

	The performance of SPIM-ISAC is only favorable only when the path gains are comparable. In~\cite{spim_bounds_JSTSP_Wang2019May}, it was shown that $\mathcal{M}_\mathrm{SPIM} > \mathcal{M}_\mathrm{mmWave}$ only if $\gamma_{1} < 4\gamma_{2}$, where $\gamma_{1} + \gamma_2 = 1$. We validate our analysis as illustrated in Fig.~\ref{fig_SE_w1}, when SNR is $20$ dB and $\eta = 0.5$. The performance of SPIM-ISAC significantly degrades as $\gamma_{1}> \gamma_{2}$, for which mmWave-ISAC always prefers the strongest path. We conclude that SPIM-ISAC is favorable when the path gains are comparable. When there is a significant gap between the path gains, then mmWave-ISAC yields higher MI. Note that the performance improvement obtained from SPIM-ISAC is limited to the number of available spatial paths in the environment. Furthermore, higher MI is achieved by employing more RF chains at the cost of higher hardware complexity.

	\begin{figure}[t]
		\centering
		{\includegraphics[draft=false,width=\columnwidth]{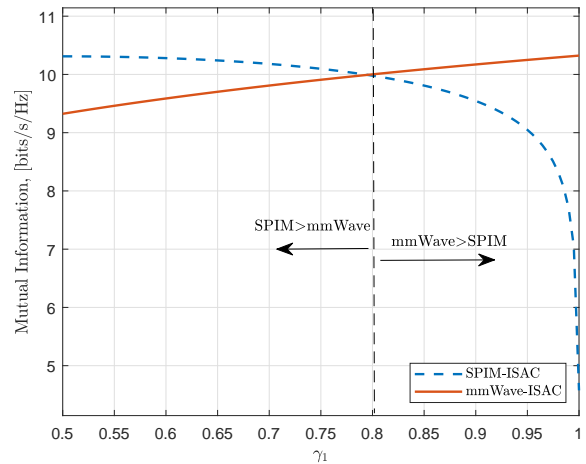} } 
		\vspace*{-6mm}
		\caption{MI versus the first path gain $\gamma_{1}$ when $\mathrm{SNR} = 20$ dB and $\eta = 0.5$. Note that $\gamma_{2} = 1-\gamma_{1}$.
		}
		
		\label{fig_SE_w1}
	\end{figure}

	\subsection{Radar Sensing Performance}
	\label{sec:RadarSensingSim}
	Next, we evaluated the MI performance with respect to the radar-communications trade-off. In Fig.~\ref{fig_BP}, the beampattern of the proposed ISAC hybrid beamformer is shown for $\eta = \{0,0.3,0.5,0.8,1\}$ and $\gamma_{1} = \gamma_{2} =0.5$, when $i = \{1,2\}$ spatial patterns are used. In this scenario, the target is located at $40^\circ$ while the BS receives the incoming paths from the communications user at $50^\circ$ ($i=1$) and $60^\circ$ ($i =2$), respectively. The beampattern becomes suppressed at the target direction when $\eta \rightarrow 1$. Conversely, the beampattern at the user locations is minimized when $\eta \rightarrow 0$. This illustrates the effectiveness of our proposed SPIM-ISAC approach.

	\begin{figure}[t]
		\centering
		{\includegraphics[draft=false,width=\columnwidth]{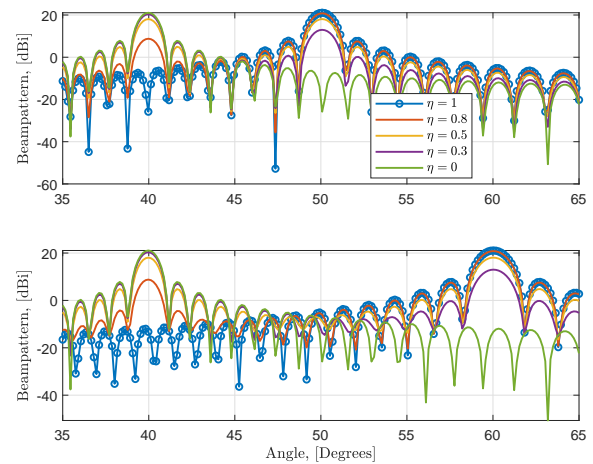} } 
		\vspace*{-4mm}
		\caption{Azimuthal beampattern when   $i=1$ (top) and $i =2$ (bottom) for various values of the radar-communications trade-off parameter, $\eta = \{0,0.3,0.5,0.8,1\}$. Here, the path gains $\gamma_{1} = \gamma_{2} =0.5$.
		}
		
		\label{fig_BP}
	\end{figure}

	\section{Summary}
	\label{sec:Conc}
	We introduced an SPIM framework for ISAC, wherein the hybrid beamformers are designed by exploiting the spatial scattering paths between the BS and the communications user. We have shown that a significant performance improvement is achieved via SPIM-ISAC compared to conventional mmWave-ISAC, wherein only the strongest path is selected for beamformer design. We have evaluated the performance of our SPIM-ISAC technique in terms of both communications (MI) and radar (beampattern) performance metrics. The proposed approach needs to be extended for a multi-target multi-user scenario. 


	\balance
	\bibliographystyle{IEEEtran}
	\bibliography{IEEEabrv,references_116}

\end{document}